\documentclass[runningheads]{llncs}
\AtBeginDocument{%
  \providecommand\BibTeX{{%
    \normalfont B\kern-0.5em{\scshape i\kern-0.25em b}\kern-0.8em\TeX}}}

\makeatletter
  \newcommand\figcaption{\def\@captype{figure}\caption}
  \newcommand\tabcaption{\def\@captype{table}\caption}
\makeatother

\usepackage{booktabs} 
\usepackage{xspace}
\usepackage{graphicx}
\usepackage{algorithm}
\usepackage[noend]{algpseudocode}
\usepackage{url}
\usepackage{color}
\usepackage{colortbl}
\usepackage{enumitem}
\usepackage{amsmath}
\usepackage{epstopdf}
\usepackage{subfigure}
\usepackage{multirow}
\usepackage{todonotes}
\usepackage{caption}
\usepackage{framed}
\usepackage{soul}
\usepackage{diagbox}

\newcommand{\HBS}[1]{#1}
\newcommand{\YC}[1]{#1}
\newcommand{\YCL}[1]{#1}
\newcommand{\QY}[1]{#1}
\newcommand{\NQY}[1]{#1}

\newcommand{\xxx}[1]{#1\xspace}

\newcommand{\sol}{{\tt DBL}\xspace}
\newcommand{\solp}{{\tt DBL-P}\xspace}
\newcommand{\solg}{{\tt DBL-G}\xspace}
\newcommand{\dl}{{\tt DL}\xspace}
\newcommand{\bl}{{\tt BL}\xspace}

\newcommand{\bfs}{{\tt BFS}\xspace}
\newcommand{\bbfs}{{\tt B-BFS}\xspace}
\newcommand{\dfs}{{\tt DFS}\xspace}

\newcommand{\dagg}{{\tt DAG}\xspace}

\newcommand{\scc}{{\tt SCC}\xspace}

\newcommand{\ip}{{\tt IP}\xspace}
\newcommand{\ipp}{{\tt IP-P}\xspace}
\newcommand{\dgr}{{\tt DAGGER}\xspace}
\newcommand{\tol}{{\tt TOL}\xspace}

\newcommand{\uconstruct}{{\tt CONSTRUCT}\xspace}

\begin{document}

\title{DBL: Efficient Reachability Queries on Dynamic Graphs (Complete Version)}


\author{Qiuyi Lyu\inst{1} \and
Yuchen Li\inst{2}\and
Bingsheng He\inst{3}\and
Bin Gong\inst{4}}
\authorrunning{Q. Lyu et al.}
%
\institute{Shandong University
\email{qiuyilv@gmail.com}\\\and
Singapore Management University
\email{yuchenli@smu.edu.sg}\\\and
National University of Singapore
\email{hebs@comp.nus.edu.sg}\\\and
Shandong University
\email{gb@sdu.edu.cn}
}

\maketitle

\begin{abstract}

Reachability query is a fundamental problem on graphs, which has been extensively studied in academia and industry. Since graphs are subject to frequent updates in many applications, it is essential to support efficient graph updates while offering good performance in reachability queries.
Existing solutions compress the original graph with the Directed Acyclic Graph (\dagg) and propose efficient query processing and index update techniques. However, they focus on optimizing the scenarios where the Strong Connected Components (\scc{s}) remain unchanged and have overlooked the prohibitively high cost of the \dagg maintenance when \scc{s} are updated. In this paper, we propose \sol, an efficient \dagg-free index to support the reachability query on dynamic graphs with insertion-only updates. \sol builds on two complementary indexes: Dynamic Landmark (\dl) label and Bidirectional Leaf (\bl) label. The former leverages landmark nodes to quickly determine reachable pairs whereas the latter prunes unreachable pairs by indexing the leaf nodes in the graph. We evaluate \sol against the state-of-the-art approaches on dynamic reachability index with extensive experiments on real-world datasets. The results have demonstrated that \sol achieves orders of magnitude speedup in terms of index update, while still producing competitive query efficiency.

\end{abstract}

\section{Introduction}\label{sec:int}
Given a graph $G$ and a pair of vertices $u$ and $v$,
reachability query (denoted as $q(u,v)$) is a fundamental graph operation that answers whether there exists
a path from $u$ to $v$ on $G$. This operation is a core component in supporting numerous applications in practice, such as those in social networks, biological complexes, knowledge graphs, and transportation networks.
A plethora of index-based approaches have been developed over a decade \cite{wei2018reachability,cohen2003reachability,wang2006dual,schenkel2005efficient,yildirim2010grail,seufert2013ferrari,su2017reachability,yildirim2013dagger,valstar2017landmark,hotz2019experiences} and
demonstrated great success in handling reachability query on \emph{static} graphs with millions of vertices and edges.
However, in many cases, graphs are highly \emph{dynamic}~\cite{wang2017real}: New friendships continuously form on social networks like Facebook and Twitter;
knowledge graphs are constantly updated with new entities and relations;
and transportation networks are subject to changes when road constructions and temporary traffic controls occur. In those applications, it is essential to support efficient graph updates while offering good performance in reachability queries.

There have been some efforts in developing reachability index to support graph updates \cite{bramandia2010incremental,demetrescu2001fully,henzinger1995fully,jin2011path,roditty2013decremental,roditty2016fully,schenkel2005efficient}.
However, there is a major assumption made in those works:
\ul{the Strongly Connected Components (\scc{s}) in the underlying graph remain unchanged after the graph gets updated.}
The Directed Acyclic Graph (\dagg) collapses the \scc{s} into vertices and the reachability query is then processed on a significantly smaller graph than the original.
The state-of-the-art solutions \cite{zhu2014reachability,wei2018reachability} thus rely on the \dagg to design an index for efficient query processing, yet their index maintenance mechanisms only support the update which does not trigger \scc merge/split in the \dagg. However, such an assumption can be invalid in practice, as edge insertions could lead to updates of the \scc{s} in the \dagg. In other words, the overhead of the \dagg maintenance has been mostly overlooked in the previous studies.

One potential solution is to adopt existing \dagg maintenance algorithms such as \cite{yildirim2013dagger}. Unfortunately, this \dagg maintenance is a prohibitively time-consuming process, as also demonstrated in the experiments. For instance, in our experiments, the time taken to update the \dagg on one edge insertion in the LiveJournal dataset is two-fold more than the time taken to process \emph{1 million queries} for the state-of-the-art methods. Therefore, we need a new index scheme with a low maintenance cost while efficiently answering reachability queries.

In this paper, we propose a \dagg-free dynamic reachability index framework(\sol) that enables efficient index update and supports fast query processing at the same time on large scale graphs. We focus on insert-only dynamic graphs with new edges and vertices continuously added. This is because the number of deletions are often significantly smaller than the number of insertions, and deletions are handled with lazy updates in many graph applications~\cite{akiba2014dynamic,boccaletti2006complex}. Instead of maintaining the \dagg, we index the reachability information around two sets of vertices: the ``landmark'' nodes with high centrality and the ``leaf'' nodes with low centrality (e.g., nodes with zero in-degree or out-degree).
As the reachability information of the landmark nodes and the leaf nodes remain relatively stable against graph updates, it enables efficient index update opportunities compared with approaches using the \dagg. Hence, \sol is built on the top of two simple and effective index components: (1) a Dynamic Landmark (\dl) label, and (2) a Bidirectional Leaf (\bl) label. Combining \dl and \bl in the \sol ensures efficient index maintenance while achieves competitive query processing performance.

\emph{Efficient query processing:}
%
\dl is inspired by the landmark index approach \cite{cohen2003reachability}.
The proposed \dl label maintains a small set of the landmark nodes as the label for each vertex in the graph. Given a query $q(u,v)$, if both the \dl labels of $u$ and $v$ contain a common landmark node, we can immediately determine that $u$ reaches $v$. Otherwise, we need to invoke Breadth-First Search(\bfs) to process $q(u,v)$. We devise \bl label to quickly prune vertex pairs that are not reachable to limit the number of costly \bfs. \bl complements \dl and it focuses on building labels around the leaf nodes in the graph. The leaf nodes form an exclusive set apart from the landmark node set. \bl label of a vertex $u$ is defined to be the leaf nodes which can either reach $u$ or $u$ can reach them. Hence, $u$ does not reach $v$ if there exists one leaf node in $u$'s \bl label which does not appear in the \bl label of $v$. \YC{In summary, \dl can quickly determine reachable pairs while \bl, which complements \dl, prunes disconnected pairs to remedy the ones that cannot be immediately determined by \dl.}

\emph{Efficient index maintenance:} Both \dl and \bl labels are lightweight indexes where each vertex only stores a constant size label. When new edges are inserted, efficient pruned \bfs is employed and only the vertices where their labels need update will be visited. In particular, once the label of a vertex is unaffected by the edge updates, we safely prune the vertex as well as its descendants from the \bfs, which enables efficient index update.

To better utilize the computation power of modern architectures, we implement \dl and \bl with simple and compact bitwise operations. Our implementations are based on OpenMP and CUDA in order to exploit parallel architectures multi-core CPUs and GPUs (Graphics Processing Units), respectively.

Hereby, we summarize the contributions as the following:

\begin{itemize}[leftmargin=1em]
	\setlength\itemsep{0em}
	\item We introduce the \sol framework which combines two complementary \dl and \bl labels to enable efficient reachability query processing on large graphs.
	\item We propose novel index update algorithms for \dl and \bl. To the best of our knowledge, this is the first solution for dynamic reachability index without maintaining the \dagg. In addition, the algorithms can be easily implemented with parallel interfaces.
	\item We conduct extensive experiments to validate the performance of \sol in comparison with the state-of-the-art dynamic methods~\cite{zhu2014reachability,wei2018reachability}. \sol achieves competitive query performance and orders of magnitude speedup for index update. We also implement \sol on multi-cores and GPU-enabled system and demonstrate significant performance boost compared with our sequential implementation.
\end{itemize}

The remaining part of this paper is organized as follows. Section~\ref{sec:pb} presents the preliminaries and background. Section~\ref{sec:relatedwork} presents the related work. Section~\ref{sec:dbl} presents the index definition as well as query processing.
Sections~\ref{sec:lan} demonstrate the update mechanism of \dl and \bl labels.
Section~\ref{sec:exp} reports the experimental results. Finally, we conclude the paper in Section~\ref{sec:con}.

\begin{table}[t]
    \scriptsize
	\centering
	\caption{Common notations in this paper}
	\vspace{-1em}
	\label{tbl:notations}
    \begin{center}
	\begin{tabular}{|c|p{300pt}|}
		\hline
		\textbf{Notation} & \textbf{Description} \\ \hline
		$G(V,E)$	& the vertex set $V$ and the edge set $E$ of a directed graph $G$  \\ \hline
		$G'$	& the reverse graph of $G$  \\ \hline
		$n$ & the number of vertex in $G$ \\ \hline
		$m$ & the number of edges in $G$ \\ \hline
		$Suc(u)$ & the set of $u$'s out-neighbors \\ \hline
        $Pre(u)$ & the set of $u$'s in-neighbors \\ \hline
		$Des(u)$ & the set of $u$'s descendants including $u$ \\ \hline
		$Anc(u)$ & the set of $u$'s ancestors including $u$ \\ \hline
		$Path(u,v)$ & A path from vertex $u$ to vertex $v$ \\ \hline
		$q(u,v)$ & the reachability query from $u$ to $v$ \\ \hline
		$k$ & the size of \dl label for one vertex \\ \hline
		$k'$ & the size of \bl label for one vertex \\ \hline
		$\dl_{in}(u)$ & the label that keeps all the landmark nodes that could reach $u$\\ \hline
		$\dl_{out}(u)$ & the label that keeps all the landmark nodes that could be reached by $u$\\ \hline
		$\bl_{in}(u)$ & the label that keeps the hash value of the leaf nodes that could reach $u$ \\ \hline
		$\bl_{out}(u)$ & the label that keeps the hash value of the leaf nodes that could be reached by $u$\\ \hline
		$h(u)$ & the hash function that hash node $u$ to a value\\ \hline
	\end{tabular}
\end{center}
\end{table}

\section{Preliminaries}\label{sec:pb}


A directed graph is defined as $G=(V,E)$, where $V$ is the vertex set and $E$ is the edge set with $n=|V|$ and $m=|E|$. 
We denote an edge from vertex $u$ to vertex $v$ as $(u,v)$.
A path from $u$ to $v$ in $G$ is denoted as $Path(u,v)=(u,w_1,w_2,w_3,\ldots,v)$ where $w_i \in V$ and the adjacent vertices on the path are connected by an edge in $G$. We say that $v$ is reachable by $u$ when there exists a $Path(u,v)$ in $G$. In addition, we use $Suc(u)$ to denote the direct successors of $u$ and the direct predecessors of $u$ are denoted as $Pre(u)$. Similarly, we denote all the ancestors of $u$ (including $u$) as $Anc(u)$ and all the descendants of $u$ (including $u$) as $Des(u)$.
We denote the reversed graph of $G$ as $G'=(V,E')$ where all the edges of $G$ are in the opposite direction of $G'$. In this paper, the forward direction refers to traversing on the edges in $G$. Symmetrically, the backward direction refers to traversing on the edges in $G'$.
We denote $q(u,v)$ as a reachability query from $u$ to $v$. In this paper, we study the dynamic scenario where edges can be inserted into the graph. Common notations are summarized in Table~\ref{tbl:notations}.

\section{Related Work}\label{sec:relatedwork}

There have been some studies on dynamic graph \cite{bramandia2010incremental,demetrescu2001fully,henzinger1995fully,jin2011path,roditty2013decremental,roditty2016fully,schenkel2005efficient}. Yildirim et al. propose \dgr \cite{yildirim2013dagger} which maintains the graph as a \dagg after insertions and deletions.
The index is constructed on the \dagg to facilitate reachability query processing.
The main operation for the \dagg maintenance is the merge and split of the \emph{Strongly Connected Component} (\scc).
Unfortunately, it has been shown that \dgr exhibits unsatisfactory query processing performance on handling large graphs (even with just millions of vertices \cite{zhu2014reachability}).

The state-of-the-art approaches: \tol \cite{zhu2014reachability} and \ip \cite{wei2018reachability} follow the maintenance method for the \dagg from \dgr and propose novel dynamic index on the \dagg to improve the query processing performance.
We note that \tol and \ip are only applicable to the scenarios where the \scc/s in the \dagg remains unchanged against updates. In the case of \scc merges/collapses, \dgr is still required to recover the \scc/s. For instance, \tol and \ip can handle edge insertions $(v_1,v_5)$ in Figure~\ref{fig:graph}, without invoking \dgr. However, when inserting $(v_9,v_2)$, two SCC/s $\{v_2\}$ and $\{v_5, v_6, v_9\}$ will be
merged into one larger SCC $\{v_2, v_5, v_6, v_9\}$. For such cases, \tol and \ip rely on \dgr for maintaining the \dagg first and then perform their respective methods for index maintenance and query processing.
However, the overheads of the \scc maintenance are excluded in their experiments \cite{zhu2014reachability,wei2018reachability} and such overheads is in fact non-negligible \cite{yildirim2013dagger,cycledetection}.

In this paper, we propose the \sol framework which only maintains the labels for all vertices in the graph without constructing the \dagg. That means, \sol can effectively avoid the costly \dagg maintenance upon graph updates.
\sol achieves competitive query processing performance with the state-of-the-art solutions (i.e., \tol and \ip) while offering orders of magnitude speedup in terms of index updates.

\begin{figure*}[t]
	\subfigure[Graph $G$]{
		\begin{minipage}{0.21\linewidth}
			\label{fig:graph}
			\includegraphics[height=3.6cm]{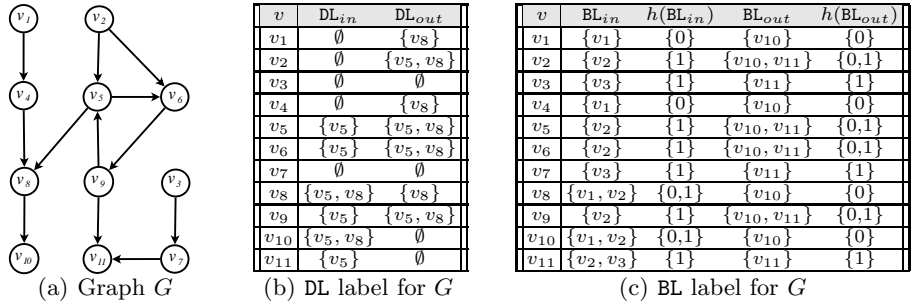}
		\end{minipage}
	}
	\hfill
	\subfigure[\dl label for $G$]{
		\begin{minipage}{0.23\linewidth}
			\centering
            \scriptsize
			\begin{tabular}{||c|cc||}
				\hline
				\rowcolor[gray]{.9}	
				\hline
				$v$	&	$\dl_{in}$	&	$\dl_{out}$	\\ \hline
				$v_1$	&	$\emptyset$	&	$\{v_8\}$				\\ \hline
				$v_2$	&	$\emptyset$	&	$\{v_5,v_8\}$				\\ \hline
				$v_3$	&	$\emptyset$ &	$\emptyset$		\\ \hline
				$v_4$	&	$\emptyset$	&	$\{v_8\}$				\\ \hline
				$v_5$	&	$\{v_5\}$			&	$\{v_5,v_8\}$				\\ \hline
				$v_6$	&	$\{v_5\}$			&	$\{v_5,v_8\}$				\\ \hline
				$v_7$	&	$\emptyset$	&	$\emptyset$		\\ \hline
				$v_8$	&	$\{v_5,v_8\}$		&	$\{v_8\}$				\\ \hline
				$v_9$	&	$\{v_5\}$			&	$\{v_5,v_8\}$				\\ \hline
				$v_{10}$	&	$\{v_5,v_8\}$		&	$\emptyset$		\\ \hline
				$v_{11}$	&	$\{v_5\}$			&	$\emptyset$		\\ \hline
			\end{tabular}
			\label{fig:dl-label}
		\end{minipage}
	}
	\hfill
	\subfigure[\bl label for $G$]{
		\begin{minipage}{0.43\linewidth}
			\centering
            \scriptsize
			\begin{tabular}{||c|cccc||}
				\hline
				\rowcolor[gray]{.9}
				\hline
				$v$	&	$\bl_{in}$	&   $h(\bl_{in})$    &	$\bl_{out}$	    &   $h(\bl_{out})$\\ \hline
				$v_1$	&	$\{v_1\}$   &	\{0\}		&	$\{v_{10}\}$		&	 \{0\}      \\ \hline
				$v_2$	&	$\{v_2\}$	&	\{1\}    	&	$\{v_{10},v_{11}\}$	&	 \{0,1\}   \\ \hline
				$v_3$	&	$\{v_3\}$ 	&	\{1\}    	&	$\{v_{11}\}$		&	 \{1\}      \\ \hline
				$v_4$	&	$\{v_1\}$	&	\{0\}    	&	$\{v_{10}\}$		&    \{0\}      \\ \hline
				$v_5$	&	$\{v_2\}$	&	\{1\}    	&	$\{v_{10},v_{11}\}$	&	 \{0,1\}    \\ \hline
				$v_6$	&	$\{v_2\}$	&	\{1\}    	&	$\{v_{10},v_{11}\}$	&    \{0,1\}	\\ \hline
				$v_7$	&	$\{v_3\}$   &	\{1\}		&	$\{v_{11}\}$	&	 \{1\}	    \\ \hline
				$v_8$	&	$\{v_1,v_2\}$	&	\{0,1\}	    &	$\{v_{10}\}$	&	 \{0\}	    \\ \hline
				$v_9$	&	$\{v_2\}$	&	\{1\}	    &	$\{v_{10},v_{11}\}$	&	\{0,1\} \\ \hline
				$v_{10}$	&	$\{v_1,v_2\}$	&   \{0,1\}     &	$\{v_{10}\}$	&	\{0\}	\\ \hline
				$v_{11}$	&	$\{v_2,v_3\}$	&   \{1\}	&	$\{v_{11}\}$	&	\{1\}	\\ \hline
			\end{tabular}
			\label{fig:bl-label}
		\end{minipage}
	}

	\caption{A running example of graph $G$}

\end{figure*}

\section{DBL Framework}\label{sec:dbl}

The \sol framework is consist of \dl and \bl label which have their independent query and update components. In this section, we introduce the \dl and \bl label. Then, we devise the query processing algorithm that builds upon \sol index.

\subsection{Definitions and Construction}\label{sec:dbl:def}
We propose the \sol framework that consists of two index components: \dl and \bl.
\begin{definition}[\dl label]
Given a landmark vertex set $L \subset V$ and $|L|=k$, we define two labels for each vertex $v \in V$: $\dl_{in}(v)$ and $\dl_{out}(v)$.
$\dl_{in}(v)$ is a subset of nodes in $L$ that could reach $v$ and $\dl_{out}(v)$ is a subset of nodes in $L$ that $v$ could reach.
\end{definition}


It is noted that \dl label is a subset of the 2-hop label \cite{cohen2003reachability}. In fact, 2-Hop label is a special case for \dl label when the landmark set $L = V$. Nevertheless, we find that maintaining 2-Hop label in the dynamic graph scenario leads to index explosion. Thus, we propose to only choose a subset of vertices as the landmark set $L$ to index \dl label. In this way, \dl label has up to $O(n|L|)$ space complexity and the index size can be easily controlled by tunning the selection of $L$. The following lemma shows an important property of \dl label for reachability query processing.

\begin{lemma}
Given two vertices $u$,$v$ and their corresponding \dl label, $\dl_{out}(u) \cap \dl_{in}(v) \neq \emptyset$ deduces $u$ reaches $v$ but not vice versa.
\end{lemma}

\begin{example}
We show an running example in Figure~\ref{fig:graph}. Assuming the landmark set is chosen as $\{v_5,v_8\}$, the corresponding \dl label is shown in Figure~\ref{fig:dl-label}. $q(v_1,v_{10})$ returns true since $\dl_{out}(v_1) \cap \dl_{in}(v_{10}) = \{v_8\}$.
However, the labels cannot give negative answer to $q(v_3,v_{11})$ despite
$\dl_{out}(v_{3}) \cap \dl_{in}(v_{11}) = \emptyset$.
This is because the intermediate vertex $v_7$ on the path from $v_3$ to $v_{11}$ is not included in the landmark set.	
\end{example}

To achieve good query processing performance, we need to select a set of vertices as the landmarks such that they cover most of the reachable vertex pairs in the graph, i.e., $\dl_{out}(u) \cap \dl_{in}(v)$ contains at least one landmark node for any reachable vertex pair $u$ and $v$. The optimal landmark selection has been proved to be NP-hard \cite{potamias2009fast}. In this paper, we adopt a heuristic method for selecting \dl label nodes following existing works  \cite{akiba2013fast,potamias2009fast}. In particular, we rank vertices with $M(u) = |Pre(u)| \cdot |Suc(u)|$ to approximate their centrality and select top-$k$ vertices. \QY{Other landmark selection methods are also discussed in Section~\ref{exp:selection}.}


\begin{definition}[\bl label]\label{def:bl}
\bl introduces two labels for each vertex $v \in V$: $\bl_{in}(v)$ and $\bl_{out}(v)$.
$\bl_{in}(v)$ contains all the zero in-degrees vertices that can reach $v$,
and $\bl_{out}(v)$ contains all the zero out-degrees vertices that could be reached by $v$.
For convenience, we refer to vertices with either zero in-degree or out-degree as the leaf nodes.
\end{definition}
\begin{lemma}\label{cor:bl}
Given two vertices $u$,$v$ and their corresponding \bl label, $u$ does not reach $v$ in $G$ if $\bl_{out}(v) \not\subseteq \bl_{out}(u)$ or $\bl_{in}(u) \not\subseteq \bl_{in}(v)$.
\end{lemma}

\bl label can give negative answer to $q(u,v)$. This is because if $u$ could reach $v$, then $u$ could reach all the leaf nodes that $v$ could reach, and all the leaf nodes that reach $u$ should also reach $v$. \dl label is efficient for giving positive answer to a reachability query whereas \bl label plays a complementary role by pruning unreachable pairs. In this paper, we take vertices with zero in-degree/out-degree as the leaf nodes. \QY{We also discuss other leaf selection methods in Section~\ref{exp:selection}.}

\begin{algorithm}[t]
	\caption{\dl label Batch Construction}
	\begin{algorithmic}[1]
		\Require
		Graph $G(V,E)$, Landmark Set $D$
		\Ensure
		$\dl$ label for $G$
		\For{$i=0$; $i < k$; $i$++ }
		\State //Forward \bfs  \label{dl:in1}
		\State $S \gets D[i]$ \label{algo:con:init}
		\State enqueue $S$ to an empty queue $Q$ \label{algo:con:start}
		\While{$Q$ not empty} \label{algo:con:include:start}
		\State $p \gets $ pop $Q$
		\For{$x \in Suc(p)$ }
		\State $\dl_{in}(x) \gets \dl_{in}(x) \cup \{S\}$;
		\State enqueue x to $Q$  \label{dl:in2}
		\EndFor
		\EndWhile \label{algo:con:include:end}
		\State //Symmetrical Backward \bfs is performed.\label{dl:out1}
		\EndFor
	\end{algorithmic}
	\label{algo:appendix:dl}
\end{algorithm}

\begin{example}
Figure~\ref{fig:bl-label} shows \bl label for the running example.
\bl label gives negative answer to $q(v_4,v_6)$ since $\bl_{in}(v_4)$ is not contained by $\bl_{in}(v_6)$.
Intuitively, vertex $v_1$ reaches vertex $v_4$ but cannot reach $v_6$ which indicates $v_4$ should not reach $v_6$.
\bl label cannot give positive answer. Take $q(v_5,v_2)$ for an example,
the labels satisfy the containment condition but positive answer cannot be given.
\end{example}

The number of \bl label nodes could be huge.
To develop efficient index operations, we build a hash set of size $k'$ for \bl as follows. Both $\bl_{in}$ and $\bl_{out}$ are a subset of $\{1, 2,\ldots, k'\}$ where $k'$ is a user-defined label size, and they are stored in bit vectors. A hash function is used to map the leaf nodes to a corresponding bit. For our example, the leaves are $\{v_1,v_2,v_3,v_{10},v_{11}\}$. When $k'=2$, all leaves are hashed to two unique values. Assume $h(v_1)=h(v_{10})=0$, $h(v_2)=h(v_3)=h(v_{11})=1$. We show the hashed \bl label set in Figure~\ref{fig:bl-label} which are denoted as $h(\bl_{in})$ and $h(\bl_{out})$. In the rest of the paper, we directly use $\bl_{in}$ and $\bl_{out}$ to denote the hash sets of the corresponding labels. It is noted that one can still use Lemma~\ref{cor:bl} to prune unreachable pairs with the hashed \bl label.

We briefly discuss the batch index construction of \sol as the focus of this work is on the dynamic scenario. The construction of \dl label is presented in Algorithm~\ref{algo:appendix:dl}, which follows existing works on 2-hop label~\cite{cohen2003reachability}.
For each landmark node $D[i]$, we start a \bfs from $S$ (Line~\ref{algo:con:start}) and include $S$ in $\dl_{in}$ label of every vertices that $S$ can reach (Lines~\ref{algo:con:include:start}-\ref{algo:con:include:end}). For constructing $\dl_{out}$, we execute a \bfs on the reversed graph $G'$ symmetrically (Line~\ref{dl:out1}). To construct \bl label, we simply replace the landmark set $D$ as the leaf set $D'$ and replace $S$ with all leaf nodes that are \emph{hashed} to bucket $i$ (Line~\ref{algo:con:init}) in Algorithm~\ref{algo:appendix:dl}.
The complexity of building \sol is that $O((k+k')(m+n))$.

Note that although we use~\cite{cohen2003reachability} for offline index construction, the contribution of our work is that we construct \dl and \bl as complementary indices for efficient query processing. Furthermore, we are the first work to support efficient dynamic reachability index maintenance without assuming \scc/s remain unchanged.



\vspace{1mm}
\noindent\textbf{Space complexity.}
The space complexities of \dl and \bl labels are $O(kn)$ and $O(k'n)$, respectively.

\begin{algorithm}[t]
	\caption{Query Processing Framework for \sol}
	\begin{algorithmic}[1]
		\Require
		Graph $G(V,E)$, \dl label, \bl label, $q(u,v)$
		\Ensure
		Answer of the query.
		\Function {$\dl\_Intersec$}{$x$,$y$}
		\State return ($\dl_{out}(x) \cap \dl_{in}(y)$);
		\EndFunction
		\Function {$\bl\_Contain$}{$x$,$y$}
		\State return ($\bl_{in}(x) \subseteq \bl_{in}(y)$ and $\bl_{out}(y) \subseteq \bl_{out}(x)$);
		\EndFunction
		
		\Procedure {Query}{$u$,$v$}
		\If{$\dl\_Intersec$($u$,$v$)}   \label{algo:framework:dl-label}
		\State return true;
		\EndIf
		\If{not $\bl\_Contain$($u$,$v$)}
		\State return false;    \label{algo:framework:bl-label}
		\EndIf
		\If{$\dl\_Intersec$($v$,$u$)}   \label{algo:framework:dl-label-1}
		\State return false;
		\EndIf
		\If{$\dl\_Intersec$($u$,$u$) or $\dl\_Intersec$($v$,$v$)}  \label{algo:framework:dl-label-2}
		\State return false;
		\EndIf
		\State Enqueue u for \bfs;
		\While{queue not empty}
		\State $w \gets $ pop queue;
		\For{vertex $ x \in Suc(w)$}
		\If{$x = v$}
		\State return true;
		\EndIf
		\If{$\dl\_Intersec$($u$,$x$)}    \label{algo:framework:dl-prune}
		\State continue;
		\EndIf
		\If{not $\bl\_Contain$($x$,$v$)}  \label{algo:framework:bl-prune}
		\State continue;
		\EndIf
		\State  Enqueue $x$;
		\EndFor
		\EndWhile
		\State \Return false;
		\EndProcedure
		
	\end{algorithmic}
	\label{algo:framework}
\end{algorithm}


\subsection{Query Processing}\label{sec:qp}

With the two indexes, Algorithm~\ref{algo:framework} illustrates the query processing framework of \sol.
Given a reachability query $q(u,v)$, we return the answer immediately if the labels are sufficient to determine the reachability (Lines \ref{algo:framework:dl-label}-\ref{algo:framework:bl-label}).
\YC{By the definitions of \dl and \bl labels, $u$ reaches $v$ if their \dl label overlaps (Line~\ref{algo:framework:dl-label}) where $u$ does not reach $v$ if their \bl label does not overlap (Line~\ref{algo:framework:bl-label}).}
Furthermore, there are two early termination rules implemented in Lines \ref{algo:framework:dl-label-1} and \ref{algo:framework:dl-label-2}, respectively.
\YC{
Line~\ref{algo:framework:dl-label-1} makes use of the properties that all vertices in a \scc contain at least one common landmark node. Line~\ref{algo:framework:dl-label-2} takes advantage of the scenario when either $u$ or $v$ share the same \scc with a landmark node $l$ then $u$ reaches $v$ if and only if $l$ appeared in the \dl label of $u$ and $v$.
}
We prove their correctness in Theorem 1 and Theorem 2 respectively. Otherwise, we turn to \bfs search with efficient pruning. The pruning within \bfs is performed as follows. Upon visiting a vertex $q$, the procedure will determine whether the vertex $q$ should be enqueued in Lines \ref{algo:framework:dl-prune} and \ref{algo:framework:bl-prune}. \bl and \dl labels will judge whether the destination vertex $v$ will be in the $Des(w)$. If not, $q$ will be pruned from \bfs
to quickly answer the query before traversing the graph with \bfs.

\begin{theorem}
In Algorithm~\ref{algo:framework}, when $\dl\_Intersec$($x$,$y$) returns false and $\dl\_Intersec$($y$,$x$) returns true, then $x$ cannot reach $y$.
\end{theorem}

\begin{proof}
$\dl\_Intersec$($y$, $x$) returns true indicates that vertex $y$ reaches $x$. If vertex $x$ reaches vertex $y$, then $y$ and $x$ must be in the same \scc(according to the definition of the \scc). As all the vertices in the \scc are reachable to each other, the landmark nodes in $\dl_{out}(y) \cap \dl_{in}(x)$ should also be included in $\dl_{out}$ and $\dl_{in}$ label for all vertices in the same \scc. This means $\dl\_Intersec$($x$, $y$) should return true. Therefore $x$ cannot reach $y$ otherwise it contradicts with the fact that $\dl\_Intersec$($x$, $y$) returns false.
\end{proof}

\begin{theorem}
In Algorithm~\ref{algo:framework}, if $\dl\_Intersec$($x$, $y$) returns false and $\dl\_Intersec$($x$,$x$) or $\dl\_Intersec$($y$,$y$) returns true then vertex $x$ cannot reach $y$.
\end{theorem}

\begin{proof}
If $\dl\_Intersec$($x$,$x$) returns true, it means that vertex $x$ is a landmark or $x$ is in the same \scc with a landmark. If $x$ is in the same \scc with landmark $l$, vertex $x$ and vertex $l$ should have the same reachability information. As landmark $l$ will push its label element $l$ to $\dl_{out}$ label for all the vertices in $Anc(l)$ and to $\dl_{in}$ label for all the vertices in $Des(l)$. The reachability information for landmark $l$ will be fully covered. It means that $x$'s reachability information is also fully covered. Thus \dl label is enough to answer the query without \bfs. Hence $y$ is not reachable by $x$ if $\dl\_Intersec(x, y)$ returns false.
The proving process is similar for the case when $\dl\_Intersec(y,y)$ returns true.
\end{proof}

\begin{algorithm}[t]
	\caption{$\dl_{in}$ label update for edge insertion}
	\begin{algorithmic}[1]
		\Require
		Graph $G(V,E)$, \dl label, Inserted edge $(u,v)$
		\Ensure
		Updated \dl label
		\If{ $\dl_{out}(u) \cap \dl_{in}(v) == \emptyset$ }
		\State Initialize an empty queue and enqueue $v$
		\While{queue is not empty}
		\State $p \gets $ pop queue
		\For{vertex $ x \in Suc(p)$}
		\If{ $\dl_{in}(u) \not\subseteq  \dl_{in}(x)$ }
		\State $\dl_{in}(x) \gets \dl_{in}(x) \cup \dl_{in}(u)$
		\State enqueue x
		\EndIf
		\EndFor
		\EndWhile
		\EndIf
	\end{algorithmic}
	\label{algo:dl-insert}
\end{algorithm}

\noindent\textbf{Query complexity.}
Given a query $q(u,v)$,
the time complexity is $O(k+k')$ when the query can be directly answered by \dl and \bl labels. 
Otherwise, we turn to the pruned \bfs search, which has a worst case time complexity of $O((k+k')(m+n))$.
Let $\rho$ denote the ratio of vertex pairs whose reachability could be directly answered by the label. The amortized time complexity is $O(\rho(k+k')+(1-\rho)(k+k')(m+n))$.
Empirically, $\rho$ is over $95\%$ according to our experiments (Table~\ref{tab:per} in Section~\ref{sec:exp}), which implies efficient query processing.

\section{DL and BL Update for Edge Insertions}\label{sec:lan}




When inserting a new edge $(u,v)$,
all vertices in $Anc(u)$ can reach all vertices in $Des(v)$.
On a high level, all landmark nodes that could reach $u$ should also reach vertices in $Des(v)$. In other words, all the landmark nodes that could be reached by $v$ should also be reached by vertices in $Anc(u)$.
\xxx{
Thus, we update the label by
1) adding $\dl_{in}(u)$ into $\dl_{in}(x)$ for all $x \in Des(v)$; and
2) adding $\dl_{out}(v)$ into $\dl_{out}(x)$ for all $x \in Anc(u)$.}

Algorithm~\ref{algo:dl-insert} depicts the edge insertion scenario for $\dl_{in}$.
We omit the update for $\dl_{out}$, which is symmetrical to $\dl_{in}$.
If \dl label can determine that vertex $v$ is reachable by vertex $u$ in the original graph before the edge insertion, the insertion will not trigger any label update (Line~1).
Lines 2-8 describe a \bfs process with pruning. For a visited vertex $x$, we prune $x$ without traversing $Des(x)$ iff $\dl_{in}(u) \subseteq \dl_{in}(x)$, because all the vertices in $Des(x)$ are deemed to be unaffected as their $\dl_{in}$ labels are supersets of $\dl_{in}(x)$.

\begin{example}
Figure~\ref{fig:example4} shows an example of edge insertion. Figure~\ref{fig:dl-label-insert} shows the corresponding $\dl_{in}$ label update process. $\dl_{in}$ label is presented with brackets. Give an edge $(v_9,v_2)$ inserted, $\dl_{in}(v_9)$ is copied to $\dl_{in}(v_2)$. Then an inspection will be processed on $\dl_{in}(v_5)$ and $\dl_{in}(v_6)$. Since $\dl_{in}(v_9)$ is a subset of $\dl_{in}(v_5)$ and $\dl_{in}(v_6)$, vertex $v_5$ and vertex $v_6$ are pruned from the \bfs. The update progress is then terminated.
\end{example}

\begin{figure}[t]
	\subfigure[Insert edge $(v_9,v_2)$ ]{
	\begin{minipage}{0.30\textwidth}
        \centering
		\label{fig:example4}
		\includegraphics[height=3.6cm]{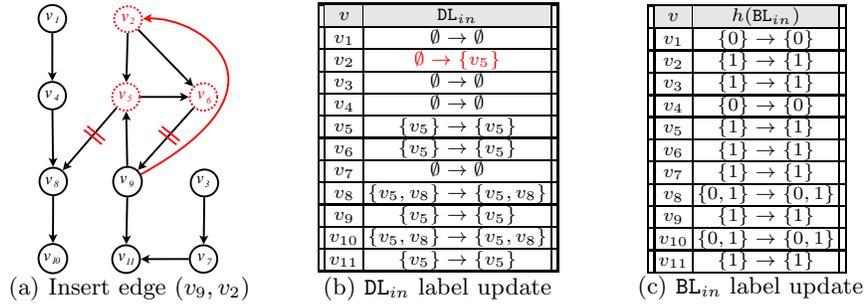}
	\end{minipage}
	}
	\hfill
	\subfigure[$\dl_{in}$ label update]{
	\begin{minipage}{0.26\textwidth}
    \scriptsize
		\centering
		\begin{tabular}{||c|c||}
			\hline
			\rowcolor[gray]{.9}	
			\hline
			$v$	    &	$\dl_{in}$		                                    \\ \hline
			$v_1$	&	$\emptyset	\rightarrow	\emptyset$	                \\ \hline
			$v_2$	&	\textcolor{red}{$\emptyset	\rightarrow	\{v_5\}$}		\\ \hline
			$v_3$	&	$\emptyset	\rightarrow	\emptyset$                 \\ \hline
			$v_4$	&	$\emptyset	\rightarrow	\emptyset$	                \\ \hline
			$v_5$	&	$\{v_5\}		\rightarrow \{v_5\}$	                    \\ \hline
			$v_6$	&	$\{v_5\}		\rightarrow	\{v_5\}$	                    \\ \hline
			$v_7$	&	$\emptyset	\rightarrow	\emptyset$	                \\ \hline
			$v_8$	&	$\{v_5,v_8\}	\rightarrow		\{v_5,v_8\}$		        \\ \hline
			$v_9$	&	$\{v_5\}		\rightarrow		\{v_5\}$	                \\ \hline
			$v_{10}$	&	$\{v_5,v_8\}	\rightarrow	\{v_5,v_8\}$	        	\\ \hline
			$v_{11}$	&	$\{v_5\}		\rightarrow	\{v_5\}$		            \\ \hline
		\end{tabular}
		\label{fig:dl-label-insert}
	\end{minipage}
	}
	\hfill
	\subfigure[$\bl_{in}$ label update]{
		\begin{minipage}{0.32\textwidth}
			\centering
            \scriptsize
			\begin{tabular}{||c|c||}
				\hline
				\rowcolor[gray]{.9}
				\hline
				$v$	&   $h(\bl_{in})$     \\ \hline
				$v_1$	  &	$\{0\}		\rightarrow	\{0\} $     \\ \hline
				$v_2$	  &	$\{1\}    	\rightarrow	\{1\} $     \\ \hline
				$v_3$	  &	$\{1\}    	\rightarrow    \{1\} $     \\ \hline
				$v_4$	  &	$\{0\}    	 \rightarrow   \{0\} $     \\ \hline
				$v_5$	  &	$\{1\}    	\rightarrow	\{1\} $     \\ \hline
				$v_6$	  &	$\{1\}    	  \rightarrow  \{1\} $	   \\ \hline
				$v_7$	  &	$\{1\}		\rightarrow	\{1\} $	   \\ \hline
				$v_8$	  &	$\{0,1\}	  \rightarrow  \{0,1\} $	   \\ \hline
				$v_9$	  &	$\{1\}	       \rightarrow	\{1\} $     \\ \hline
				$v_{10}$	&   $\{0,1\}   \rightarrow	\{0,1\}	$   \\ \hline
				$v_{11}$	&   $\{1\}		\rightarrow\{1\}	$   \\ \hline
			\end{tabular}
			\label{fig:bl-label-insert}
		\end{minipage}
	}

	\caption{Label update for inserting edge $(v_9,v_2)$}
	\label{exp:dl-insert}
	
\end{figure}

$\dl$ label only gives positive answer to a reachability query. In poorly connected graphs, $\dl$ will degrade to expensive \bfs search.
Thus, we employ the Bidirectional Leaf (\bl) label to complement $\dl$ and quickly identify vertex pairs which are not reachable. We omit the update algorithm of \bl, as they are very similar to those of \dl, except the updates are applied to $\bl_{in}$ and $\bl_{out}$ labels. Figure~\ref{fig:bl-label-insert} shows the update of $\bl_{in}$ label. Similar to the \dl label, the update process will be early terminated as the $\bl_{in}(v_2)$ is totally unaffected after edge insertion. Thus, no $\bl_{in}$ label will be updated in this case.

\vspace{1mm}
\noindent\textbf{Update complexity of \sol.}
In the worst case, all the vertices that reach or are reachable to the updating edges will be visited.
\YCL{Thus, the time complexity of \dl and \bl is $O((k+k')(m+n))$ where $(m+n)$ is the cost on the \bfs. }
Empirically, as the \bfs procedure will prune a large number of vertices, the actual update process is much more efficient than a plain \bfs.

\begin{table*}[t]
	\centering
	\caption{Dataset statistics}
	\label{tab:databsets}
	\scriptsize

	\begin{tabular}{ccccccccc}
		\hline
		\multirow{3}{*}{Dataset} & \multirow{3}{*}{$|V|$}  & \multirow{3}{*}{$|E|$}  & \multirow{3}{*}{$d_{avg}$}   &  \multirow{3}{*}{\textbf{Diameter}}& \multirow{3}{*}{\textbf{Connectivity}} &
		\multirow{3}{*}{\dagg-$|V|$}  & \multirow{3}{*}{\dagg-$|E|$}  & \dagg \\
		& & & & &      & & & \uconstruct \\
		& & & & & (\%) & & & (ms) \\
		\hline
		LJ        &4,847,571  &68,993,773 &14.23  &16     &78.9   &971,232    &1,024,140  &2368\\
		Web       &875,713    &5,105,039  &5.83   &21     &44.0   &371,764    &517,805    &191\\
		Email     &265,214    &420,045    &1.58   &14     &13.8   &231,000    &223,004    &17\\
		Wiki      &2,394,385  &5,021,410  &2.09   &9      &26.9   &2,281,879  &2,311,570  &360\\
		BerkStan  &685,231    &7,600,595  &11.09  &514    &48.8   &109,406    &583,771    &1134\\
		Pokec     &1,632,803  &30,622,564 &18.75  &11     &80.0   &325,892    &379,628    &86\\
		Twitter   &2,881,151  &6,439,178  &2.23   &24     &1.9    &2,357,437  &3,472,200  &481\\
		Reddit    &2,628,904  &57,493,332 &21.86  &15     &69.2   &800,001    &857,716    &1844\\
		\hline
	\end{tabular}
 
\end{table*}




\section{Experimental Evaluation}\label{sec:exp}

In this section, we conduct experiments by comparing the proposed \sol framework with the state-of-the-art approaches on reachability query for dynamic graphs.


\subsection{Experimental Setup}

\noindent\textbf{Environment:}
Our experiments are conducted on a server with an Intel Xeon CPU E5-2640 v4 2.4GHz, 256GB RAM and a Tesla P100 PCIe version GPU.

\noindent\textbf{Datasets:}
We conduct experiments on 8 real-world datasets (see Table~\ref{tab:databsets}).
We have collected the following datasets from SNAP \cite{snapnets}. LJ and Pokec are two social networks, which are power-law graphs in nature. BerkStan and Web are web graphs in which nodes represent web pages and directed edges represent hyperlinks between them. Wiki and Email are communication networks. Reddit and Twitter are two social network datasets obtained from \cite{wang2017real}.

\vspace{-2em}
\begin{table}[h]
	\centering
	\caption{\QY{Query time (ms) for different landmark nodes selection. A=$\max(|Pre(\cdot)|,|Suc(\cdot|))$; B=$\min(|Pre(\cdot)|,|Suc(\cdot|))$;
		C=$|Pre(\cdot)|+|Suc(\cdot)|$; D is the betweenness centrality; ours=$|Pre(\cdot)|\cdot|Suc(\cdot)|$}}
	\label{tab:landmark-rank}
	\begin{tabular}{cccccc}
		\hline
		Dataset  & A & B & C & D & ours  \\
		\hline
		LJ        &125.10      &127.84    &105.88     &113.34     &108.51  \\
		Web       &202.16      &144.13    &142.16     &140.79     &139.64     \\
		Email     &37.02       &37.01     &36.14      &38.53      &36.38   \\
		Wiki      &156.21      &159.74    &153.66     &155.45     &157.12          \\
		Pokec     &37.69       &64.57     &36.96      &50.66      &34.78   \\
		BerkStan  &1890        &6002      &1883       &1252       &1590  \\
		Twitter   &719.31      &849.78    &685.31     &727.59     &693.71  \\
		Reddit    &99.21       &65.06     &62.68      &69.62      &60.48  \\
		\hline
	\end{tabular}

\end{table}

\subsection{Label Node Selection}\label{exp:selection}
\QY{\dl select the landmark nodes by heuristically approximating the centrality of a vertex $u$ as $M(u) = |Pre(u)| \cdot |Suc(u)|$. Here, we evaluate different heuristic methods for landmark nodes selection. The results are shown in Table~\ref{tab:landmark-rank}. Overall, our adopted heuristic ($|Pre(u)| \cdot |Suc(u)|$) achieves the performance. For Email and Wiki, all the methods share a similar performance. $|Pre(\cdot)|+|Suc(\cdot)|$ and $|Pre(\cdot)|\cdot|Suc(\cdot)|$ get a better performance in other datasets. Finally, $|Pre(\cdot)|+|Suc(\cdot)|$(degree centrality) and $|Pre(\cdot)|\cdot|Suc(\cdot)|$ deliver similar performance for most datasets and the latter is superior in the Berkstan dataset.
Thus, we adopt $|Pre(\cdot)|\cdot|Suc(\cdot)|$ for approximating the centrality. It needs to mention that, although the betweenness centrality get a medium overall performance, it shows the best performance in BerkStan dataset.}

\QY{In the main body of this paper, we restrict the leaf nodes to be the ones with either zero in-degree or zero out-degree. Nevertheless, our proposed method does not require such a restriction and could potentially select any vertex as a leaf node. Following the approach for which we select \dl label nodes, we use  $M(u) = |Pre(u)| \cdot |Suc(u)|$ to approximate the centrality of vertex $u$ and select vertex $u$ as a \bl label node if $M(u) \leq r$ where $r$ is a tunning parameter. Assigning $r=0$ produces the special case presented in the main body of this paper. The algorithms for query processing as well as index update of the new \bl label remains unchanged. Figure~\ref{fig:bl-degree} shows the query performance of \sol when we vary the threshold $r$. With a higher $r$, more vertices are selected as the leaf nodes, which should theoretically improve the query processing efficiency. However, since we employ the hash function for \bl label, more leaf nodes lead to higher collision rates. This explains why we don't observe a significant improvement in query performance.}

\begin{figure}[htb]
	\centering
	\includegraphics[width=0.5\textwidth,height=5cm]{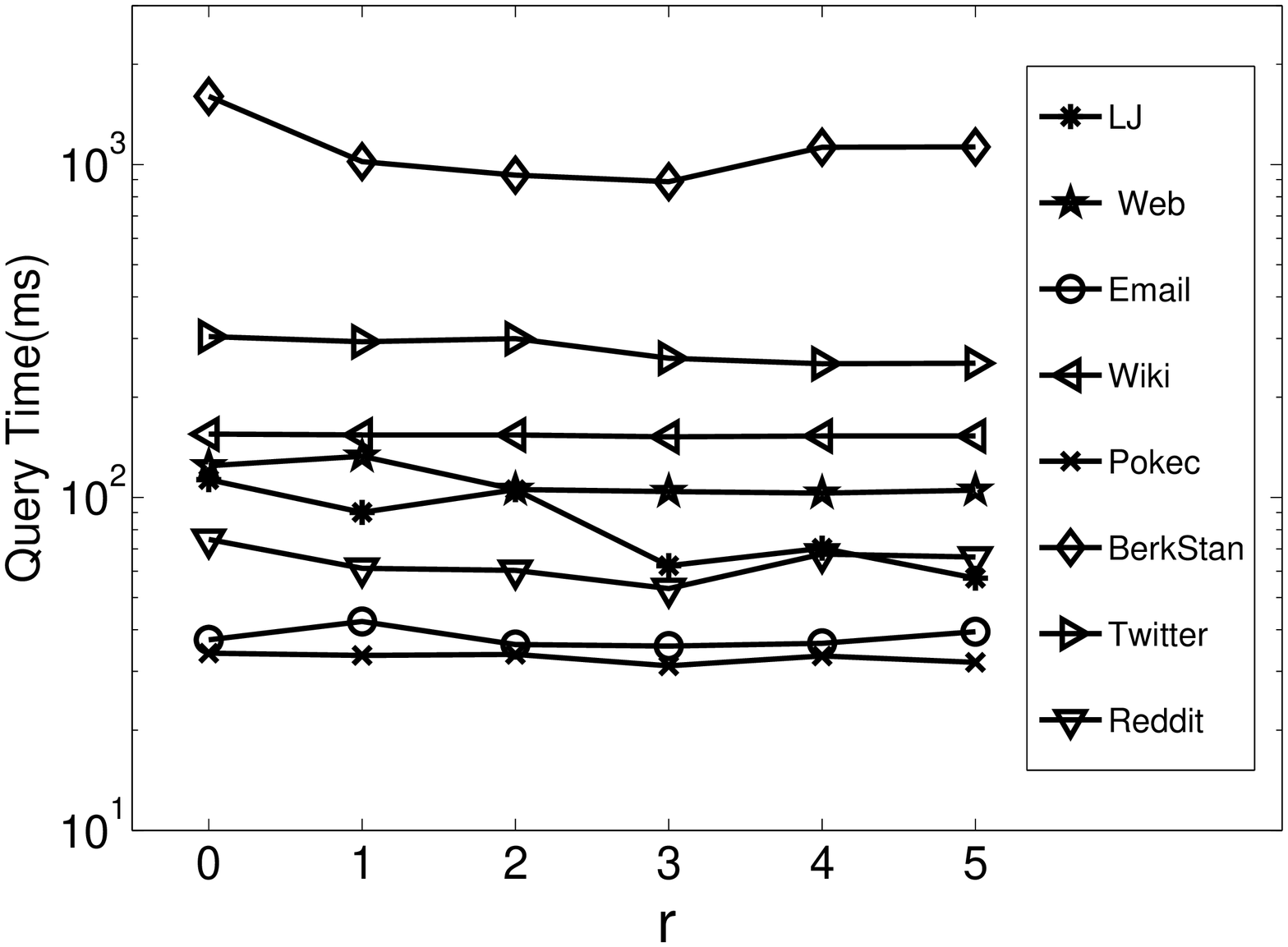}
	\caption{\QY{\bl label node selection}}
	\label{fig:bl-degree}
\end{figure}
\vspace{-2em}

\subsection{Effectiveness of DL+BL}\label{sec:exp:parameters}
Table~\ref{tab:percentage} shows the percentages of queries answered by \dl label, \bl label (\emph{when the other label is disabled}) and \sol label. All the queries are randomly generated. The results show that \dl is effective for dense and highly connected graphs (LJ, Pokec and Reddit) whereas \bl is effective for sparse and poorly connected graphs (Email, Wiki and Twitter). However, we still incur the expensive \bfs if the label is disabled. By combining the merits of both indexes, our proposal leads to a significantly better performance. \NQY{\sol could answer much more queries than \dl and \bl label}. The results have validated our claim that \dl and \bl are complementary to each other.
\YC{We note that the query processing for \sol is able to handle one million queries with sub-second latency for most datasets, which shows outstanding performance.}

\noindent\textbf{Impact of Label Size:} On the query processing of \sol.
There are two labels in \sol: both \dl and \bl store labels in bit vectors.
The size of \dl label depends on the number of selected landmark nodes whereas the size of \bl label is determined by how many hash values are chosen to index the leaf nodes. We evaluate all the datasets to show the performance trend of varying \dl and \bl label sizes per vertex (by processing 1 million queries) in Table~\ref{tab:bldl}.

When varying \dl label size $k$, the performance of most datasets remain stable before a certain size (e.g., 64) and deteriorates thereafter. This means that extra landmark nodes will cover little extra reachability information. Thus, selecting more landmark nodes does not necessarily lead to better overall performance since the cost of processing the additional bits incur additional cache misses. BerkStan gets benefit from increasing the \dl label size to 128 since 64 landmarks are not enough to cover enough reachable pairs.

\begin{table}[t]
		\captionof{table}{Percentages of queries answered by \dl label, \bl label (when the other label is disabled) and \sol label respectively. We also include the time for \sol to process 1 million queries}
        \centering
        \label{tab:per}
		{
			\begin{tabular}{ccccc}
				\hline
				Dataset  & \dl Label & \bl Label & \sol Label & \sol time  \\
				\hline
				LJ      &97.5\%    &20.8\%    &99.8\%  	&  	108ms	\\
				Web     &79.5\%    &54.3\%    &98.3\% 	&	139ms	\\
				Email   &31.9\%    &85.4\%    &99.2\% 	&	36ms	\\
				Wiki    &10.6\%    &94.3\%    &99.6\% 	&	157ms	\\
				Pokec   &97.6\%    &19.9\%    &99.9\% 	&	35ms	\\
				BerkStan    &87.5\%    &43.3\%  &95.0\% &	1590ms	\\
				Twitter &6.6\%     &94.8\%    &96.7\% 	&	709ms	\\
				Reddit  &93.7\%    &30.6\%    &99.9\% 	&	61ms	\\
				\hline
			\end{tabular}\label{tab:percentage}
		}
\end{table}

Compared with \dl label, some of the datasets get a sweet spot when varying the size of \bl label. This is because there are two conflicting factors which affect the overall performance. With increasing \bl label size and more hash values incorporated, we can quickly prune more unreachable vertex pairs by examining \bl label without traversing the graph with \bfs. Besides, larger \bl size also provides better pruning power of the \bfs even if it fails to directly answer the query (Algorithm~\ref{algo:framework}). Nevertheless, the cost of label processing increases with increased \bl label size. According to our parameter study, we set Wiki's \dl and \bl label size as 64 and 256, BerkStan's \dl and \bl label size as 128 and 64. For the remaining datasets, both \dl and \bl label sizes are set as 64.

    \begin{table*}[b]
    \vspace{-2em}
    \scriptsize
    \subfigure[Varying \bl label sizes]{
    \begin{minipage}{0.43\textwidth}
	\label{tab:bl}
	\begin{tabular}{|c|c|c|c|c|c|}
		\hline
		Dataset &16 &32 &64 &128 & 256 \\
        \hline
		LJ    &136.1     &131.9   &108.1   &107.4     &110.3     \\
		Web   &177.2     &128.5   &152.9   &156.6     &174.3        \\
		Email &77.4      &53.9    &38.3    &41.1      &44.4       \\
		Wiki  &911.6     &481.4   &273.7   &181.3     &157.4     \\
		Pokec &54.8      &43.7    &38.6    &40.6      &53.6      \\
		BerkStan&4876.1  &4958.9    &4862.9    &5099.1    &5544.3   \\
		Twitter&1085.3   &845.7     &708.2     &652.7     &673.2  \\
		Reddit&117.1     &80.4      &67.3      &63.5      &67.9  \\
		\hline
	\end{tabular}
    \end{minipage}
    }
	\hfill
    \subfigure[Varying \dl label sizes]{
    \begin{minipage}{0.47\textwidth}
	\centering
	\label{tab:dl}
	\begin{tabular}{|c|c|c|c|c|c|}
		\hline
		Dataset &16 &32 &64 &128 & 256 \\
        \hline
		LJ    &108.2     &110.3   &106.9   &120.2     &125.5     \\
		Web   &154.0     &152.5   &151.1   &158.8     &167.8        \\
		Email &37.9      &39.5    &35.8    &39.8      &43.7       \\
		Wiki  &274.5     &282.6   &272.4   &274.8     &281.1     \\
		Pokec &38.1      &40.6    &36.3    &49.7      &55.6      \\
		BerkStan&6369.8  &5853.1    &4756.3    &1628.3    &1735.2   \\
		Twitter&716.1   &724.4    &695.3     &707.1     &716.9  \\
		Reddit&64.6     &65.9      &62.9      &75.4      &81.4  \\
		\hline
	\end{tabular}
    \end{minipage}
    }
    \caption{Query performance(ms) with varying \dl and \bl label sizes}
    \label{tab:bldl}
    \end{table*}

\subsection{General Graph Updates} \label{sec:exp:general}

In this section, we evaluate \sol's performance on general graph update. As \dgr is the only method that could handle general update, we compare \sol against \dgr in Figure~\ref{fig:insdel}. Ten thousand edge insertion and 1 million queries are randomly generated and performed, respectively. Different from \dgr, \sol don't need to maintain the \dagg, thus, in all the datasets, \sol could achieve great performance lift compared with \dgr. For both edge insertion and query, \sol is orders of magnitude faster than \dgr. The minimum performance gap lies in BerkStan. This is because BerkStan has a large diameter. As \sol rely on \bfs traversal to update the index. The traversal overheads is crucial for it's performance. BerkStan's diameter is large, it means, during index update, \sol need to traversal extra hops to update the index which will greatly degrade the performance.

\begin{figure}[t]
\centering
    \includegraphics[width=0.9\linewidth]{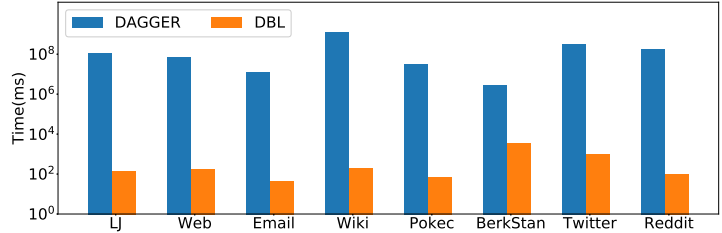}
	\caption{The execution time for insert 10000 edges as well as 1 million queries}
	\label{fig:insdel}
\end{figure}

\subsection{Synthetic graph updates} \label{sec:exp:compare}
In this section, we compare our method with \ip and \tol. Different from \sol, which could handle real world update, \ip and \tol could only handle synthetic edge update that will not trigger \dagg maintaining. Thus, for \ip and \tol, we follow their experimental setups depict in their paper\cite{zhu2014reachability,wei2018reachability}.
Specifically, we randomly select 10,000 edges from the \dagg and delete them. Then, we will insert the same edges back. In this way, we could get the edge insertion performance without trigger \dagg maintenance. For \sol, we stick to general graph updates. The edge insertion will be randomly generated and performed. One million queries will be executed after that. It needs to be noted that, although both \ip and \tol claim they can handle dynamic graph, due to their special pre-condition, their methods are in fact of limited use in real world scenario.

\HBS{The results are shown in Figure~\ref{fig:sccdi}. \sol outperforms other baselines in most cases except on three data sets (Wiki, BerkStan and Twitter) where \ip could achieve a better performance. Nevertheless, \sol outperforms \ip and \tol by 4.4x and 21.2x, respectively with respect to geometric mean performance.}
We analyze the reason that \sol can be slower than \ip on Wiki, BerkStan and Twitter. As we aforementioned, \sol relies on the pruned \bfs to update the index, the \bfs traversal speed will determine the worst-case update performance. Berkstan has the largest diameter as 514 and Twitter has the second largest diameter as 24, which dramatically degrade the update procedure in \sol. For Wiki, \sol could still achieve a better update performance than \ip. However, \ip is much more efficiency in query processing which lead to better overall performance.

Although this experimental scenario has been used in previous studies, the comparison is unfair for \sol. As both \ip and \tol rely on the \dagg to process queries and updates, \textbf{their synthetic update exclude the \dagg maintaining procedure/overheads from the experiments. }However, \dagg maintenance is essential for their method to handle real world edge updates, as we have shown in Figure~\ref{fig:insdel}, the overheads is nonnegligible.

\begin{figure}[t]
\centering
    \includegraphics[width=0.9\linewidth]{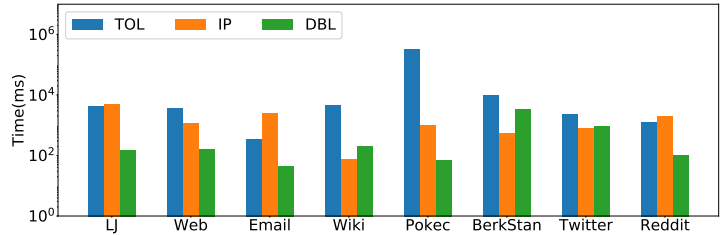}
	\caption{The execution time for insert 10000 edges as well as 1 million queries, for \tol and \ip, the updates are synthetic that will not trigger \scc update}
	\label{fig:sccdi}
\end{figure}

\subsection{Parallel Performance}
We implement \sol with OpenMP and CUDA (\solp and \solg respectively)
to demonstrate the deployment on multi-core CPUs and GPUs achieves encouraging speedup for query processing.
We follow existing GPU-based graph processing pipeline by batching the queries and updates~\cite{guo2017parallel,sha2017accelerating,li2021gpu}.
Note that the transfer time can be overlapped with GPU processing to minimize data communication costs.
Both CPU and GPU implementations are based on the vertex centric framework.


To validate the scalability of the parallel approach, we vary the number of threads used in \solp and show its performance trend in Figure~\ref{fig:scalability}.
\solp achieves almost linear scalability against increasing number of threads (note that the y-axis is plotted in log-scale).
The linear trend of scalability tends to disappear when the number of threads is beyond 14. We attribute this observation as the memory bandwidth bound nature of the processing tasks.
\sol invokes the \bfs traversal once the labels are unable to answer the query and the efficiency of the \bfs is largely bounded by CPU memory bandwidth.
This memory bandwidth bound issue of CPUs can be resolved by using GPUs which provide memory bandwidth boost.


The compared query processing performance is shown in Table~\ref{tab:gpus}. Bidirectional \bfs(\bbfs) query is listed as a baseline. We also compare our parallel solutions with a \HBS{home-grown} OpenMP implementation of \ip (denoted as \ipp). Twenty threads are used in the OpenMp implementation. We note that \ip has to invoke a pruned \dfs if its labels fail to determine the query result.
\dfs is a sequential process in nature and cannot be efficiently parallelized.
For our parallel implementation \ipp, we assign a thread to handle one query. We have the following observations.

First, \sol is built on the pruned BFS which can be efficiently parallelized with the vertex-centric paradigm. We have observed significant performance improvement by parallelized executions. \solp (CPUs) gets 4x to 10x speedup across all datasets. \solg (GPUs) shows an even better performance. In contrast, as \dfs incurs frequent random accesses in \ipp, the performance is bounded by memory bandwidth. Thus, parallelization does not bring much performance gain to \ipp compared with its sequential counterpart.

Second, \sol provides competitive efficiency against \ipp but \sol can be slower than \tol and \ip when comparing the single thread performance.
However, this is achieved by assuming the \dagg structure but the \dagg-based approaches incur prohibitively high cost of index update, as we demonstrated in the previous subsections. In contrast, \sol achieves sub-second query processing performance for handling 1 million queries while still support efficient updates without using the \dagg.

Third, there are cases where \solp outperforms \solg, i.e., Web, Berkstan and Twitter. This is because these datasets have a higher diameter than the rest of the datasets and the pruned \bfs needs to traverse extra hops to determine the reachability. Thus, we incur more random accesses, which do not suit the GPU architecture.

\begin{figure*}[t]
\begin{minipage}[h]{0.5\linewidth}
	\centering
	\includegraphics[width=1\linewidth]{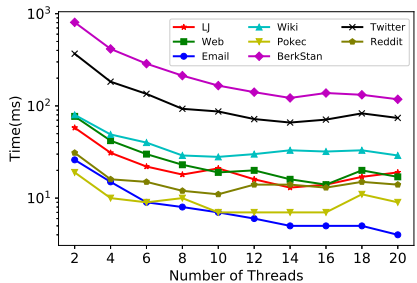}
	\vspace{-3em}
    \figcaption{Scalability of \sol on CPU}
	\label{fig:scalability}

\end{minipage}
\hfill
\begin{minipage}[h]{0.50\textwidth}
    \scriptsize
	\centering
    \vspace{2em}
	\begin{tabular}{cccccccc}
		\hline
		Dataset  &\tol &\ip &\ipp & \sol & \solp & \solg &\bbfs \\
		\hline
		LJ&46.6&50.7&24.9&108.1&16.4&\textbf{6.1}&555561\\
		Web&40.6&39.7&22.6&139.2&\textbf{12.4}&14.2&236892 \\
		Email&26.8&21.6&9.4&36.4&4.1&\textbf{2.8}&10168 \\
		Wiki&74.9&12.7&\textbf{4.1}&157.2&28.4&14.8&61113 \\
		Pokec&27.2&37.6&23.2&34.8&9.0&\textbf{3.1}&253936 \\
		BerkStan&37.2&31.6&\textbf{16.4}&1590.0&131.0&835.1&598127 \\
		Twitter&64.6&30.2&\textbf{7.1}&709.1 &79.4 &202.1&78496 \\
		Reddit&56.7&44.7&19.56&61.2 &14.6 &\textbf{3.1}&273935 \\
		\hline
	\end{tabular}
\tabcaption{The query performance(ms) on CPU and GPU architectures. B-BFS means the bidirectional BFS}
	\label{tab:gpus}
\end{minipage}

\end{figure*}

\section{Conclusion}\label{sec:con}
In this work, we propose \sol, an indexing framework to support dynamic reachability query processing on incremental graphs.
To our best knowledge, \sol is the first solution which avoids maintaining \dagg structure to construct and build reachability index.
\sol leverages two complementary index components: \dl and \bl labels.
\dl label is built on the landmark nodes to determine reachable vertex pairs that connected by the landmarks, whereas \bl label prunes unreachable pairs by examining their reachability information on the leaf nodes in the graph. The experimental evaluation has demonstrated that the sequential version of \sol outperforms the state-of-the-art solutions with orders of magnitude speedups in terms of index update while exhibits competitive query processing performance. The parallel implementation of \sol on multi-cores and GPUs further boost the performance over our sequential implementation. As future work, we are interested in extending DBL to support deletions, which will be lazily supported in many applications. 

\noindent\textbf{Acknowledgement.}
Yuchen Li’s work was supported by the Ministry of Education, Singapore, under its Academic Research Fund Tier 2 (Award No.:
MOE2019-T2-2-065).

\bibliographystyle{splncs04}
\bibliography{reference}

\end{document}